\documentclass{amsart}

\usepackage[foot]{amsaddr}
\usepackage[utf8]{inputenc}
\usepackage[T1]{fontenc}
\usepackage{amssymb,latexsym}
\usepackage{amsmath}
\usepackage{amsthm}
\usepackage{amssymb}
\usepackage{graphicx}
\usepackage{color}
\usepackage[shortlabels]{enumitem}
\usepackage[colorlinks=true,citecolor={blue},draft=false, urlcolor={black}]{hyperref}
\usepackage{cleveref}
\usepackage{todonotes}
\usepackage{amssymb,epsfig}
\usepackage{a4}
\usepackage{graphics}

\def\A{\mathcal A}
\def\B{\mathcal B}
\def\D{\mathcal D}

\def\Q{\mathcal Q}
\def\W{\mathcal W}
\def\C{\mathbb C}
\def\Z{\mathbb Z}
\def\R{\mathbb R}
\def\N{\mathbb N}

\newtheorem{thm}{Theorem}
\newtheorem{theorem}[thm]{Theorem}

\newtheorem{corollary}[thm]{Corollary}
\newtheorem{lem}[thm]{Lemma}

\newtheorem{proposition}[thm]{Proposition}

\newtheorem{defi}[thm]{Definition}

\crefname{thm}{theorem}{theorems}
\crefname{theorem}{theorem}{theorems}
\crefname{coro}{corollary}{corollaries}
\crefname{example}{example}{examples}
\crefname{lem}{lemma}{lemmas}
\crefname{lmm}{lemma}{lemmas}
\crefname{claim}{claim}{claims}
\crefname{obs}{observation}{observations}
\crefname{proposition}{proposition}{propositions}
\crefname{prop}{proposition}{propositions}
\crefname{defi}{definition}{definitions}

\newtheorem{remark}[thm]{Remark}

\newtheorem{example}[thm]{Example}
\crefname{example}{example}{examples}
\newcommand{\norm}[1]{\left\lVert#1\right\rVert} 

\author[I.~Farkas]{Izabella Ingrid Farkas}
\email{ingrid.farkas@inf.elte.hu}
\address{Eötvös Loránd Univerity in Budapest, Hungary}

\author[E.~Pelantov\'a]{Edita Pelantov\'a}
\email{edita.pelantova@fjfi.cvut.cz}
\address{FNSPE Czech Technical University in Prague, Czech Republic}

\author[M.~Svobodov\'a]{Milena Svobodov\'a}
\email{milenasvobodova@volny.cz}
\address{FNSPE Czech Technical University in Prague, Czech Republic}

\title{From positional representation of numbers to positional representation of vectors}

\begin{document}


\begin{abstract}

To represent real $m$-dimensional vectors, a positional vector system given by a non-singular matrix $M \in \Z^{m \times m}$ and a digit set $\D \subset \Z^m$ is used.  If $m = 1$, the system coincides with the well known numeration system used to represent real numbers. We study  some properties of the vector systems which are transformable from the case $m = 1$ to higher dimensions. We focus on algorithm for parallel addition and on systems allowing an eventually periodic representation of vectors with rational coordinates.

\end{abstract}

\maketitle

\section{Introduction}

Expression  of a number as a linear combination of elements of the sequence $(\beta^j)_{j \in \Z}$ with coefficients from a finite set~$\D$ is nowadays the most used way how to represent numbers. Such a system is called a positional number system with the base $\beta$ and the digit set~$\D$. The decimal number system with the base ten and digits $0, 1, \ldots, 9$ prevails in Europe for several centuries. In the age of computers, the binary and hexadecimal number systems broke the domination of the decimal number system. But advantages of working with another number system  were  observed still before computers came on the scene: A.~Cauchy~\cite{Ca1840} checked correctness of his computation using simultaneously the classical decimal number system and decimal number system with the (symmetric and redundant) set of digits $\{-5, -4, \ldots, 0, 1, \ldots, 5\} $. V.~Gr\"unwald~\cite{Gr1885} considered a number system with base $\beta = -2$ and digits $\{ 0, 1\}$. An important moment for the number systems, making them a source of interest for many areas of mathematics, came in 1957, when A.~R\'enyi~\cite{Re57} introduced number systems with an arbitrary real base $\beta > 1$. Algebraic, dynamical, topological, geometric and  algorithmic properties of the R\'enyi number systems have been very intensively studied since then, from both theoretical and practical points of view. For example, a suitable choice of a base in an algebraic extension of rational numbers enables to represent all elements of the algebraic field by a finite or eventually periodic string of digits, see \cite{Sch80} and~\cite{VaVe18}. Further generalisations of numeration systems emerged in the following years. Knuth~\cite{Kn60} and Penney~\cite{Pen65} came with positional representations of complex numbers, wherein, instead of two strings of digits representing the real and the imaginarys part of complex numbers separately, they suggested to use a complex base~$\beta$, in order to represent the complex number by a single string of (real) digits. This type of representation was then further developed in the concept of canonical number systems~\cite{KoPe91} and (even more general) shift radix systems, see~\cite{KiTh14} for a survey on the topic.

In most of the above mentioned generalisations of numeration systems, every number has a unique representation, whose digits are determined by iterations of some transformation function. One exception is the decimal system with symmetric digit set used by Cauchy. It has eleven digits - more than necessary (for representing all positive integers). Such a system is called redundant. Already Cauchy noticed that, with this redundant system, the addition of two numbers is easier than in the classical system, as the carry propagation is limited. This property was further exploited by A.~Avizienis, aiming to speed up addition. In the classical $b$-ary numeration system, where the base is an integer $\beta = b \geq 2$, addition has linear time complexity with respect to the length of representations of the summands. Avizienis~\cite{Avizienis} designed an algorithm with constant time complexity for addition in redundant number systems using base $\beta = b \geq 3$ and a symmetric (integer) digit set.

In this paper, we consider representations of  $m$-dimensional vectors. The numeration system is given by a (square matrix) base $M \in \Z^{m \times m}$ and by a finite set of digits $\D \subset \Z^m$. An origin of such numeration systems can be found in works \cite{Vince93a} and~\cite{Vince93b} of A.~Vince, showing that for any expansive matrix $M \in \Z^{m \times m}$ there exists a digit set $\D \subset \Z^m$ such that any element~$x$ of the lattice~$\Z^m$ can be written in the form $x = \sum_{j=0}^n M^j {d}_j$, where $d_j \in \D$. In other words, the whole lattice~$\Z^m$ is representable in the matrix numeration system $(M, \D)$. On the other hand, if a matrix~$M$ has an eigenvalue inside the unit circle, then no choice of the digit set $\D \subset \Z^m$ allows to represent all integer vectors as a combination of only non-negative powers of~$M$. The matrix formalism for numeration systems (under the name numeration systems in lattices) was systematically used by A.~Kov\'acs in \cite{KovacsAttila03}. Of course, Kov\'acs, just like Vince, considers integer matrices, as they map a lattice into itself. In fact, already positional representations of Gaussian integers or algebraic integers from an algebraic extension of rational numbers  can be interpreted as special cases of the matrix numeration systems. J.~Jankauskas and J.~Thuswaldner generalised the Vince's results to matrix bases $M \in \mathbb{Q}^{m \times m}$ with rational entries and without eigenvalues in modulus strictly smaller than~$1$, see~\cite{JaTh18}. Another generalisation introduced recently allows to use both positive and negative powers of the matrix base for representation of vectors. In~\cite{PelVa22}, it is shown that for $M \in \Z^{m \times m}$ with $\det M = \Delta \neq 0$ there exists a finite digit set $\D \subset \Z^m$ such that every integer vector from~$\Z^m$ has a finite $(M, \D)$-representation, i.e.,
\begin{equation}\label{Fin}
    \Z^m \subset {\rm Fin}_{\D}(M) := \Bigl\{ \sum_{j \in I} M^j {d}_j: I \text{ finite subset of } \Z \, , {d}_j \in \D \Bigr\} \, .
\end{equation}
We show (in Theorem~\ref{veta}) that if, moreover, no eigenvalue of~$M$ lies on the unit circle, then for a suitable (finite) digit set $\D \subset \Z^m$, addition and subtraction on ${\rm Fin}_{\D}(M)$ can be performed by a parallel algorithm, i.e., in a constant number of steps independent of the length of $(M, \D)$-representation of summands. According to Proposition~\ref{Ano}, the required assumption on eigenvalues of~$M$ is in fact necessary for existence of a parallel addition algorithm on $(M, \D)$. Then we restrict our study to expansive matrices -- i.e., matrices with all eigenvalues strictly outside the unit circle. In Theorem~\ref{periodic1}, we show that the digit set $\D \subset \Z^m$ allowing parallel addition enables (for an expansive matrix base~$M$) eventually periodic $(M, \D)$-representation of every element of~$\mathbb{Q}^m$, i.e.
\begin{equation*}
    \mathbb{Q}^m = {\rm Per}_{\D}(M): = \Bigl\{ \sum_{j=-\infty}^N M^j {d}_j : N \in \N \, , \ d_j \in \D \text{ for each } j \leq N \text{ and } (d_j)_{-\infty}^N \text{ is eventually periodic } \Bigr\} \, .
\end{equation*}
Consequently, every element of~$\R^m$ has an $(M, \D)$-representation (Corollary~\ref{total}).

The methods we use in our proofs are based on proofs of analogous results for positional representation of real and complex numbers, modified accordingly to the formalism of matrices and vectors.


\section{Preliminaries}

A numeration system used for positional representation of complex numbers is given by a base $\beta \in \C$ with $|\beta| > 1$ and a finite digit set $\A \subset \C$.  If $x \in \C$ can be written in the form $x = \sum_{j=-\infty}^n a_j \beta^j$, where $a_j \in \A$ for each $j \in \Z, j \leq n$, we say that~$x$ has a $(\beta, \A)$-representation. The assumption $|\beta| > 1$ guarantees that the series $\sum_{j=-\infty}^n a_j \beta^j$ is convergent for any choice of digits $a_j \in \A$.
W.~Penney in~\cite{Pen65} introduced the following numeration system, which we use to demonstrate our approach.
\begin{example}\label{Ex:Penney}
Let us consider $\beta = i-1$. Penney in~\cite{Pen65} showed that
\begin{enumerate}
    \item each $x \in \Z[i] = \{a + \imath b: a, b \in \Z\}$ can be expressed uniquely as $x = \sum_{j=0}^n a_j \beta^j$, where $a_0, a_1, \ldots, a_n \in \{0, 1\}$ and $a_n \neq 0$ (if $x \neq 0$);
    \item each $x \in \C$ can be expressed as $x = \sum_{j=-\infty}^n a_j \beta^j$, where $a_j \in \{0, 1\}$ for every $j \in \Z, j \leq n$.
\end{enumerate}
\end{example}

Our aim is to study selected properties of matrix numeration systems used to represent $m$-dimensional vectors. Any matrix numeration system used in this paper is given by a non-singular matrix base $M \in \Z^{m \times m}$ and a finite (vector) digit set~$\D \subset \Z^m$. Thanks to the result of~\cite{PelVa22} mentioned earlier, we always assume that
\begin{eqnarray}
    \label{Zm-in-Fin}
     && \Z^m {\rm \ is\ a\ subset\ of\ {\rm Fin}_{\D}(M) \ as\ introduced\ in~\eqref{Fin},\ and,\ moreover,}\\
    \label{0-in-D}
     && \D {\rm \ contains\ the\ zero\ vector.}
\end{eqnarray}

In the first part of this paper, we work only with vectors from ${\rm Fin}_{\D}(M)$. Therefore, we do not yet impose additional assumptions on the matrix base~$M$, analogous to the assumption $|\beta| > 1$ required for (complex) number bases (which is important to ensure convergence of the infinite series $\sum_{j=-\infty}^n a_j \beta^j$). Only in the second part of the paper, we revisit the question of infinite representations convergence for matrix numeration systems as well.

\medskip

Let us list some obvious properties of the set ${\rm Fin}_{\D}(M)$:
\begin{itemize}
    \item $M^p {\rm Fin}_{\D}(M) = {\rm Fin}_{\D}(M)$ for every $p \in \Z$.
    \item ${\rm Fin}_{\D}(M) \subset \mathbb{Q}^m$, more precisely, ${\rm Fin}_{\D}(M) \subset  \bigcup\limits_{k \in \N} \frac{1}{\Delta^k} \Z^{m} $, where $\Delta = \det M$.
    \item If $x = \sum_{j=0}^n\ {M^j d_j}$ for some $n \in \N$, then $x \in \Z^m$.
    \item ${\rm Fin}_{\D}(M)$ is closed under addition and subtraction: Indeed, if $x, y \in {\rm Fin}_{\D}(M)$, then there exists $p\in \N$ such that  $M^p x \in \Z^m$ and $M^p y \in \Z^m$, and hence $M^p (x \pm y) \in \Z^m$. By assumption \eqref{Zm-in-Fin}, $M^p (x \pm y) \in {\rm Fin}_{\D}(M) = M^p {\rm Fin}_{\D}(M)$, and thus $x \pm y \in {\rm Fin}_{\D}(M)$.
\end{itemize}

\medskip

If a vector $x \in \mathbb{Q}^m$ is expressed as $x  =\sum_{j\in I} M^j d_j$ for a finite $I \subset \Z$ and $d_j \in \D$, we can, for some integer numbers $s\leq 0 \leq n$, write $x = \sum_{j=s}^{n}\ M^j d_j$, because the zero vector belongs to $\D$. Hence $x$ can be identified with a bi-infinite string $(d_j)_{j \in \Z} \in \D^\Z$ usually referred to as $(M,\D)$-representation of $x$:
\begin{equation*}
    (x)_{M, \D} =  {}^{\omega}0 d_n d_{n-1} \cdots d_1 d_0 \bullet d_{-1} d_{-2} \cdots d_{s} 0^\omega \, ,
\end{equation*}
where the zero index in the bi-infinite string is indicated by $\bullet$.\\

As mentioned in the introduction, some numeration systems used for representation of numbers can also be interpreted as matrix numeration systems. Let us illustrate this concept on the Penney numeration system introduced in Example~\ref{Ex:Penney}.

\begin{example}\label{Ex:Penney-number-vs-matrix} Let us transform the number numeration system from Example~\ref{Ex:Penney} into a matrix numeration system on  $Z^2$.  In place of the number base $\beta = \imath-1$, we use the matrix base $M = \left( \begin{array}{cc} -1 & -1 \\ +1 & -1 \end{array} \right) \in \Z^{2 \times 2}$.

It is easily seen that multiplication of a Gaussian integer -- complex number $x = b + \imath c \in \Z[\imath]$, with $b, c \in \Z$, by the (number) base $\beta$ corresponds to multiplication of an integer vector $v = (b, c)^\top \in \Z^2$ by the (matrix) base $M$, as follows:
\begin{eqnarray*}
    \beta x & = & (\imath-1) \cdot (b + \imath c) = (-b-c) + \imath (b-c) \, , \\
    M v & = & \left( \begin{array}{cc} -1 & -1 \\ +1 & -1 \end{array} \right)
        \cdot \left( \begin{array}{c} b \\ c \end{array} \right)
        = \left( \begin{array}{c} -b-c \\ b-c \end{array} \right) \, .
\end{eqnarray*}

Let us define a mapping $\xi: \Z[\imath] \mapsto \Z^2$ by the formula
\begin{equation}\label{isomorph}
\xi (b + \imath c) := (b, c)^\top \quad {\mathrm for \ any \ } b, c \in \Z \, .
\end{equation}
Obviously, $ \xi (x + y) = \xi (x) + \xi (y)$ for every   $x, y \in \Z[\imath]$. Thus the mapping $\xi$ is an isomorphism between the lattices $\Z[\imath]$ and  $\Z^2$. Moreover, it  fulfils the equality
\begin{equation}
\xi (\beta \cdot x) = M \cdot \xi (x) \quad \text{ for every $x \in \Z[\imath]$}.
\end{equation}
Hence, if $b + \imath c \in \Z[\imath]$ is written in the form $b + \imath c = \sum_{j=0}^n a_j \beta^j $ with $a_j \in \{0, 1\}$, then
$$\left( \begin{array}{c} b \\ c \end{array} \right) = \xi(b + \imath c ) = \xi\Bigl(\sum_{j=0}^n \beta^j a_j \Bigr) = \sum_{j=0}^n M^j d_j, \quad \text{where} \quad d_j =\xi(a_j) \in \left\{ \left( \begin{array}{c} 0 \\ 0 \end{array} \right),\left( \begin{array}{c} 1 \\ 0 \end{array} \right)\right\} = \left\{ \xi(0), \xi(1) \right\} \, .$$

Using the properties of the Penney numeration system from Example~\ref{Ex:Penney}, we conclude that the matrix numeration system given by the base $M = \left( \begin{array}{cc} -1 & -1 \\ +1 & -1 \end{array} \right)$ and the digit set $\D = \{ (0, 0)^\top, (1, 0)^\top \}$ provides for any vector $v \in \Z^2$ a unique $(M,\D)$-representation in the form $v = \sum_{j=0}^n M^j d_j$, $d_j \in \D$ and $d_n \neq 0$ (if $v \neq 0$).

\end{example}

\section{Parallel addition in matrix numeration systems}

Let us consider the operations of addition and subtraction on the set of $m$-dimensional vectors from algorithmic point of view. Similarly to the classical algorithms for arithmetic operations, we work only with finite representations -- i.e., on the set ${\rm Fin}_{\mathcal{D}}(M)$. Let $x, y
\in {\rm Fin}_{\mathcal{D}}(M)$, with
\begin{equation*}
    (x)_{M, \D} = {}^{\omega}0 x_n x_{n-1} \cdots x_1 x_0 \bullet x_{-1} x_{-2} \cdots x_{s} 0^\omega \quad \text{ and } \quad (y)_{M, D} = {}^{\omega}0 y_n y_{n-1} \cdots y_1 y_0 \bullet y_{-1} y_{-2} \cdots y_{s} 0^\omega .
\end{equation*}
Adding $x$ and~$y$ means to rewrite the $(M, \D + \D)$-representation
\begin{equation*}
    ^{\omega}0 (x_n + y_n) \cdots (x_1 + y_1) (x_0 + y_0) \bullet (x_{-1} + y_{-1}) \cdots (x_{s} + y_{s}) 0^\omega
 \end{equation*}
of the number $x+y$ into an $(M, \D)$-representation of~$x+y$.

As already announced, we are interested in parallel algorithms for addition. Let us mathematically formalise the parallelism. Firstly, we recall the notion of a {\em local function}, which comes from
 symbolic dynamics, see~\cite{LM}.

\begin{defi}\label{local}
Let $\A$ and $\B$ be finite sets. A function $\varphi : \A^{\Z} \rightarrow \B^{\Z}$ is said to be {\em $p$-local} if there exist non-negative integers $r$ and~$t$ satisfying $p = r + t + 1$, and a
function $\Phi : \A^p \rightarrow \B$ such that, for any $u = (u_j)_{j \in \Z} \in \A^{\Z}$ and its image $v = \varphi(u) = (v_j)_{j\in \Z} \in \B^{\Z}$, we have $v_{j} = \Phi(u_{j+t} \cdots u_{j-r})$ for every $j \in \Z$.
\end{defi}

This means that the image of~$u$ by~$\varphi$ is obtained through a window of limited length~$p$. The parameter~$r$ is called \emph{memory} and the parameter~$t$ is called \emph{anticipation} of the function~$\varphi$. Such functions, restricted to finite sequences, are computable by a parallel algorithm in constant time, irrespective of the length of the operands' representations.

\begin{defi}\label{digitsetconv}
Given a (matrix) base $M \in \Z^{m \times m}$ with $\det M \neq 0$ and (vector) digit sets $\A, \B \subset \Z^m$ containing~$0$, a {\em digit set conversion} in base~$M$ from~$\A$ to~$\B$ is a function $\varphi : \A^\Z \rightarrow \B^\Z$ such that
\begin{enumerate}
    \item for any $u = (u_j)_{j \in \Z} \in \A^\Z$ with a finite number of non-zero digits, $v = (v_j)_{j\in \Z} = \varphi(u) \in \B^{\Z}$ has only a finite number of non-zero digits, and
    \item $\sum\limits_{j \in \Z} M^j v_j = \sum\limits_{j \in \Z} M^j u_j$.
\end{enumerate}
Such a conversion is said to be {\em computable in parallel} if it is a~$p$-local function for some $p \in \N$.
\end{defi}

Thus, the operation of addition on ${\rm Fin}_{\D}(M)$ is computable in parallel if there exists a digit set conversion in base~$M$ from $\D + \D$ to $\D$ which is computable in parallel.

Two useful lemmas precede the statement about parallel addition on matrix numeration systems:

\begin{lem}\label{extendDigits}
Let $M \in \Z^{m \times m}$ be a non-singular matrix and $\D \subset \Z^m$ be a finite digit set such that every $x \in \Z^m $ is representable in the numeration system $(M, \D)$. If addition is computable in parallel in $(M, \D)$, then it is computable in parallel  also in $(M, \D')$ for each finite digit set $\D' \subset \Z^m$ containing $\D$.
\end{lem}

\begin{proof}
Each digit $d' \in \D'$ can be written in the form $d' = \sum_{j \in I(d')} M^j d_j$, where $I(d')$ is a finite subset of~$\Z$ and $d_j \in \D$ for each $j \in I(d')$. Let $q = \max\{|x| : d' \in \D', x \in I(d') \}$.

The string $0^\omega d'_n d'_{n-1} \cdots d'_0 \bullet d'_{-1} d'_{-2} \cdots d'_{-N} 0^\omega$  can be transformed by a~$(2q+1)$-local function into a string $0^\omega e_{n+q} e_{n+q-1} \cdots e_0 \bullet e_{-1} e_{-2} \cdots e_{-N-q} 0^\omega$, where $e_d \in \underbrace{\D + \D + \cdots + \D}_{(2q+1)-\text{times}}$.

In other words, a sum of two finite $(M, \D')$-representations can be rewritten as sum of $2(2q+1)$~finite $(M, \D)$-representations. Since addition of two strings is doable in parallel in $(M, \D)$, addition of $2(2q+1)$~strings (with fixed $q$) is possible in parallel $(M, \D)$ as well, and the resulting $(M, \D)$-representation is also an $(M, \D')$-representation, due to $\D \subset \D'$.
\end{proof}

The following lemma is stated without proof here, as the course of the proof would be identical to that of Proposition~5.1 in~\cite{FrPeSv}. Although that proposition works with roots of the minimal polynomial of an algebraic number, the minimality of the polynomial is not used for the proof itself. In fact, the idea of the proof comes from~\cite{AkDrJa12}, where expansive polynomials are considered. In our case, we extend the considerations to polynomials with no roots on the unit circle.

\begin{lem}\label{dominant}
Let $\alpha_1, \ldots, \alpha_n \in \C$ be the roots of a~polynomial $f \in \Z[X]$ satisfying $|\alpha_k| \neq 1$ for all $k = 1, \ldots, n$. Then for any $t \geq 1$ there exists a non-zero polynomial $g \in \Z[X]$, $g(X) = \sum_{l=0}^{p-1} c_l X^l$ such that $g$ is divisible by $f$ and for one coefficient $c_L$ we have
\begin{equation}
    \frac{~1~}{t} \ c_L > \sum_{l = 0, l \neq L}^{p-1} |c_l| \, .
\end{equation}
In particular, if $|\alpha_k| > 1$ for all $k = 1, \ldots, n$, then $L=0$.
\end{lem}

\begin{corollary}\label{inequal} Let $M \in \Z^{m \times m}$, with $\det M \neq 0$ and no eigenvalue of~$M$ equals $1$~in modulus.  Then there exists a polynomial $g\in \Z[X]$, $g(X) = \sum_{l=0}^{p-1} c_l X^l$ such that $g(M) = \Theta$ and for one coefficient $c_L$ we have
\begin{equation}\label{less}
    c_L \geq 4 \cdot \Bigl(C + \sum_{l=0, l \neq L}^{p-1} |c_l| \Bigr) \,  \quad \text{where}\ \ C= c_L - 4 \bigl\lfloor \frac{c_L}{4} \bigr\rfloor \in \{0, 1, 2, 3\} \, .
\end{equation}
\end{corollary}

\begin{proof}
Let $f$ be the characteristic polynomial of the matrix~$M$. The Hamilton--Cayley theorem says that $f(M) = \Theta$. By Lemma~\ref{dominant} applied on $f$ with $t = 16$, we find a polynomial $g(X) = \sum_{l=0}^{p-1} c_l X^l \in \Z[X]$ such that
\begin{equation*}
    c_L > 16 \cdot \sum_{l=0, l \neq L}^{p-1} |c_l| \, , \text{ for some } L \in \{0, \ldots, p-1 \} \, .
\end{equation*}
Using the fact that $\sum_{l=0, l \neq L}^{p-1} |c_l| \geq 1$ and denoting $C := c_L - 4 \bigl\lfloor \frac{c_L}{4} \bigr\rfloor \in \{0, 1, 2, 3\} \, $,
we obtain the following estimate:
$$c_L > 16 \cdot \sum_{l=0, l \neq L}^{p-1} |c_l|
= 4 \cdot \Bigl(3 \sum_{l=0, l \neq L}^{p-1} |c_l| + \sum_{l=0, l \neq L}^{p-1} |c_l| \Bigr)
\geq 4 \cdot \Bigl(3 + \sum_{l=0, l \neq L}^{p-1} |c_l| \Bigr)
\geq 4 \cdot \Bigl(C + \sum_{l=0, j \neq L}^{p-1} |c_l| \Bigr) \, .$$
Since the characteristic polynomial $f$ divides~$g$, we have $g(M) = \Theta$.
\end{proof}

\begin{example}\label{pol} The minimal polynomial of the complex number $\beta = \imath - 1$ and the characteristic polynomial of the matrix~$M$ defined in Example~\ref{Ex:Penney-number-vs-matrix} are both equal to $f(X) = X^2 + 2X + 2$. The polynomial $g(X) = X^4 +4 = (X^2+2X +2)(X^2- 2X+2)$ satisfies $g(M) = \Theta$ and \eqref{less} with $L=0$.
\end{example}

\begin{thm}\label{veta}
Let $M \in \Z^{m \times m}$, with $\det M \neq 0$ and no eigenvalue of~ $M$ equals $1$~in modulus. Then there exists a finite (vector) digit set $\D \subset \Z^{m}$ such that $\Z^{m} \subset {\rm Fin}_{\D}(M)$ and both addition and subtraction on ${\rm Fin}_{\D}(M)$ are computable in parallel.
\end{thm}

\begin{proof} Let $g(X) = \sum_{l=0}^{p-1} c_l X^l$ be the polynomial from Corollary \ref{inequal}. Denote $K = \bigl\lfloor \frac{c_L}{4} \bigr\rfloor \geq 1$ and define $\D = [-3K, 3K)^m \cap \Z^m$. In order to show that the digit set~$\D$ enables parallel addition, we introduce two auxiliary sets
\begin{equation*}
    \D' = [-2K, 2K)^m \cap \Z^m \qquad \text{and} \qquad \Q = [-1,1]^m \cap \Z^m \, ,
\end{equation*}
and then exploit the obvious fact that
\begin{equation}\label{InIt}
    \D + \D \subset \D' + 4K \Q \, .
\end{equation}
Let $x = \sum_{j \in \Z} M^j a_j$ and $y = \sum_{j \in \Z} M^j b_j$, where $a_j, b_j \in \D$. Moreover, we assume that $a_j, b_j \neq 0$ for just a finite number of indices $j \in \Z$. Clearly, $a_j + b_j \in \D + \D$. Due to~\eqref{InIt}, we find for each~$j$ a vector $q_j \in \Q$ such that $a_j + b_j - 4K q_j \in \D'$. Then
\begin{equation*}
    x + y = \sum_{j \in \Z} M^j (a_j + b_j) = \sum_{j \in \Z} \Bigl(M^j (a_j + b_j) - \underbrace{M^{j-L} g(M)}_{=\Theta} q_j \Bigr) \, .
\end{equation*}
We express~$g(M)$ in the explicit polynomial form $g(M) = \sum_{l=0}^{p-1} c_l M^l$ in the rightmost sum:
\begin{equation*}
    \sum_{j \in \Z} {M^{j-L} g(M)} q_j = \sum_{j \in \Z} \sum_{l=0}^{p-1} M^{j-L+l} c_l q_j = \sum_{j \in \Z} M^j \Bigl(\sum_{l=0}^{p-1} c_l q_{j+L-l} \Bigr) \, .
\end{equation*}
Therefore, $x + y = \sum_{j \in \Z} {M}^j z_j$, with
\begin{equation}\label{sum}
    z_j = a_j + b_j - \Bigl(\sum_{l=0}^{p-1} c_l q_{j+L-l} \Bigr) \, .
\end{equation}
Using $c_L = 4K + C$, we get
\begin{equation*}
    z_j = \underbrace{a_j + b_j - 4K {q_j}}_{\in \D'} - \Bigl(\underbrace{C {q_j} + \sum_{l=0, l \neq L}^{p-1} c_l {q}_{j+L-l}}_{=: {u}}\Bigr) \, .
\end{equation*}
All entries of all vectors~$q_j$ belong to $\{0, 1, -1\}$, and thus any component of the vector~$u$ in modulus is at most $C + \sum_{l=0, l \neq L} |c_l|$. Equation \eqref{less} guarantees that the components of~$u$ are not greater than $\frac{c_L}{4}$. Since the components are integer, they are at most $K = \bigl\lfloor \frac{c_L}{4} \bigr\rfloor$. As $\D' + [-K,K]^m \subset [-3K,3K)^m$, we can conclude that ${z}_j \in \D$.

In order to compute ${z}_j$, we needed to know, besides ${a}_j$ and~${b}_j$, also ${q}_{j+L}, {q}_{j+L-1}, \ldots, {q}_{j+L-p+1}$. Let us stress that~${q}_j$ depends only on ${a}_j + {b}_j$. Hence, ${z}_j$ is determined by digits on $p$ positions, i.e., the addition is performed by a $p$-local function.

To demonstrate the point (1) of Definition \ref{digitsetconv}, we have to show that  $z_j \neq 0$ for only finitely many indices $j \in \Z$. The form of $\D$, $\D'$ and $\Q$ guarantees that there exists a unique $q_j$ satisfying $a_j + b_j - 4K q_j \in \D'$. In particular, if $a_j = b_j = 0$, then $q_j =0$. The formula \eqref{sum} implies that  $z_j$ is non-zero for only finitely many indices $j \in \Z$. Let us note that the  digit set~$\D$ is not closed under multiplication by~$-1$. But for each $b \in \D$ we can find $c, d \in \D$ such that $-b = c + d$. Hence subtraction of two vectors $x = \sum_\Z M^j a_j$ and $y = \sum_\Z M^j b_j$ can be viewed as addition of three vectors, and therefore it is computable in parallel as well.

\medskip

It remains to prove that  $\mathbb{Z}^m \subset {\rm Fin}_{\D}(M)$. But this is clear, since $\mathcal{D} \subset {\rm Fin}_{\D}(M)$, $ {\rm Fin}_{\D}(M)$ is closed under addition and each $x\in \Z^m$ can be expressed   as a finite sum of digits from $\mathcal{D}$.
\end{proof}

\begin{remark}\label{FractionalPart}
The vectors we add by the parallel algorithm as described in the previous proof are represented by both-sided infinite strings. But only finitely many entries of the strings are occupied by non-zero digits. Assume that $x + y = \sum_{j=n}^N M^j ({a_j} + {b_j})$, for some integers $n\leq N$.
As stated in the previous proof, if both digits ${a}_j$ and~${b}_j$ are zero, then the algorithm puts $q_j = 0$. The formula~\eqref{sum} for~${z_j}$ implies that ${z_j}$~is zero for all $j \leq n-L-1$ and for all $j \geq N+p-L$. Hence,
$\sum_{j=n}^N M^j ({a_j} + {b_j}) = \sum_{j=n'}^{N'} M^j z_j$, where $n'= n-L$ and $N' = N+p-L-1$.
\end{remark}

\begin{example} Consider the matrix numeration system with base matrix $M = \left( \begin{array}{cc} -1 & -1 \\ +1 & -1 \end{array} \right) \in \Z^{2 \times 2}$. By Example~\ref{pol}, the polynomial $g(X) = X^4 +4$ with $c_L = c_0 = 4$ is suitable for the parallel addition algorithm as described in the proof of Theorem~\ref{veta}. Following the proof, we put $K = \bigl\lfloor \frac{c_L}{4} \bigr\rfloor = 1$ and define $\D = [-3, 3)^2 \cap \Z^2$, i.e., the digit set has 36 elements. With such a choice of the digit set $\D$, addition in $(M, \D)$ is computable in parallel.
\end{example}

The algorithm for parallel addition constructed in the proof of Theorem \ref{veta} is very simple, as the value  $q_j$ depends only on the digits $a_j$ and $b_j$ having the same index $j$. An algorithm with such property is usually called {\it neighbour free}. However, we pay a large price for the simplicity of the algorithm -- the digit set is huge.
With another choice of the algorithm, the digit set could be substantially smaller, and still sufficient to perform addition in parallel by means of a $p$-local function (with a larger parameter $p$, though).

\begin{example}\label{Ex:Penney-matrix_5-digits} Consider the numeration system in $\C$ with base $\beta = \imath-1$. In \cite{LegSvo}, a $7$-local function of parallel addition in system $(\beta, \A)$ is found for the digit set $\A = \{0, \pm 1, \pm \imath\}$.  Let us denote the $7$-local function as $\varphi : (\A+\A)^7 \mapsto \A$, acting on a $7$-tuple $(w_j, \ldots, w_{j-6}) \in (\A+\A)^7$ by means of an auxiliary quotient function $Q : (\A+\A)^6 \mapsto \Q \subset \Z[\imath]$ as follows:
\begin{eqnarray}
    \label{ParAd-b1}
    q_j & := & Q(w_j, \ldots, w_{j-5}) \in \Q \quad {\mathrm \ and, \ consequently,} \\
    \label{ParAd-b2}
    z_j & := & w_j + q_{j-1} - \beta q_j = \varphi (w_j, \ldots, w_{j-6}) \in \A \, .
\end{eqnarray}

The local functions $\varphi$ and $Q$ acting on numbers can be transformed to local functions $\varphi': (\D+\D)^7 \mapsto \D = \xi(\A)$ and $Q': (\D+\D)^6 \mapsto \Q' = \xi(\Q)$ acting on vectors, by means of the isomorphism $\xi : \Z[\imath] \mapsto \Z^2$ defined in Example~\ref{Ex:Penney-number-vs-matrix}. Thereby, we obtain the following functions:
\begin{eqnarray}
    \label{ParAd-M1}
    q'_j & := & Q'(w'_j, \ldots, w'_{j-5}) := \xi (Q(\xi^{-1}(w'_j), \ldots, \xi^{-1}(w'_{j-5}))) \in Q' \\
    \label{ParAd-M2}
    z'_j & := & w'_j + q'_{j-1} - M q'_j = \varphi'(w'_j, \ldots, w'_{j-6}) =: \xi (\varphi(\xi^{-1}(w'_j), \ldots, \xi^{-1}(w'_{j-6}))) \in \D \, .
\end{eqnarray}

It means that, with help of the formulas \eqref{ParAd-b1} and \eqref{ParAd-b2} from $7$-local parallel addition on the number system $(\beta, \A)$, we obtain $7$-local parallel addition on the matrix system $(M, \D)$ by the formulas \eqref{ParAd-M1} and \eqref{ParAd-M2}, with digit set size $\#\D = \#(\xi(\A)) = \#\A = 5$. The vector digit set of size~$5$ has elements $\{(0, 0)^\top, (1, 0)^\top, (-1, 0)^\top, (0, 1)^\top, (0, -1)^\top \} = \D$.

As proved in~\cite{Leg}, the size of $5$ is minimal for a digit set allowing  parallel addition on the number system  with  base $\beta = \imath-1$.  Consequently, the digit set size $\#\D = 5$ must be minimal for parallel addition on the matrix system $(M, \D)$ as well, due to the isomorphism~$\xi$.
\end{example}

\medskip

The algorithm for parallel addition of vectors in $\Z^2$ presented in the previous Example~\ref{Ex:Penney-matrix_5-digits} uses, for the given matrix base $M$, a digit set of the minimal possible size for parallel addition. However, the way to determine the coefficients~$q_j$ is very laborious, as the formula~\eqref{ParAd-M1} is in fact a look up table with $13^6$~rows. With digit set size increased from~$5$ to~$9$ elements, a lot simpler algorithm for parallel addition  can be obtained, as presented in the following Example~\ref{Ex:Penney-matrix_9-digits}.

\begin{example}\label{Ex:Penney-matrix_9-digits}  Let us consider  $M = \left( \begin{array}{cc} -1 & -1 \\ +1 & -1 \end{array} \right)$ and the digit set
\begin{equation}
    \tilde{\D} = \{ (b, c)^\top \, : \,  b, c \in \{0, \pm 1\} \} \subset \Z^2 \quad {\mathrm of \ size \ } \#\tilde{\D} = 9 \, .
\end{equation}

Again, we construct an auxiliary coefficient function $\tilde{Q} : (\tilde{\D} + \tilde{\D})^2 \mapsto \tilde{\Q} \subset \Z^2$. The coefficients $\tilde{q}_j \in \tilde{\Q}$ produced by $\tilde{Q}$ then provide the result sum digits $\tilde{z}_j \in \tilde{\D}$ via local function $\tilde{\varphi} : (\tilde{\D} + \tilde{\D})^3 \mapsto \tilde{\D}$, as follows:
\begin{eqnarray}
    \label{ParAd-M1_9-digits}
    \tilde{q}_j & := & \tilde{Q}(\tilde{w}_j, \tilde{w}_{j-2}) \in \tilde{\Q} \quad {\mathrm \ and, \ consequently,} \\
    \label{ParAd-M2_9-digits}
    \tilde{z}_j & := & \tilde{w}_j + \tilde{q}_{j-2} - M^2 \tilde{q_j} = \tilde{\varphi} (\tilde{w}_j, \tilde{w}_{j-2}, \tilde{w}_{j-4}) \in \tilde{\D} \, .
\end{eqnarray}

The coefficient set $\tilde{\Q}$ is, just by coincidence, equal to the digit set $\tilde{\D}$:
\begin{equation*}
    \tilde{\Q} = \{ (0, 0)^\top, \pm(1, 0)^\top, \pm(0, 1)^\top, \pm(1, 1)^\top, \pm(1, -1)^\top \} \, .
\end{equation*}

The interim sum digit set $\tilde{\W} = (\tilde{\D} + \tilde{\D})$ has $25$ elements, with $(\pi/2)$-rotation symmetry given by rotation matrix $R_{\pi/2} = \left( \begin{array}{cc} 0 & -1 \\ +1 & 0 \end{array} \right)$:
\begin{eqnarray*}
    \tilde{\W} & = & (\tilde{\D} + \tilde{\D}) = \{ (b, c)^\top \, | \,  b, c \in \{0, \pm 1, \pm 2\} \}
                = \tilde{\W}_0 \cup \tilde{\W}_1 \cup \tilde{\W}_2 \cup \tilde{\W}_3 \, , \quad {\mathrm where} \\
    \tilde{\W}_k & = & (R_{\pi/2})^k \cdot \tilde{\W}_0 = (R_{\pi/2})^k \cdot \{ (0, 0)^\top, (1, 0)^\top, (2, 0)^\top, (1, 1)^\top, (2, 1)^\top, (1, 2)^\top, (2, 2)^\top \} \, .
\end{eqnarray*}

Thanks to the $(\pi/2)$-rotation symmetry of all the sets in question, i.e., $\tilde{\D}$, $\tilde{\Q}$, and $\tilde{W} = (\tilde{\D} + \tilde{\D})$, it is enough to specify the coefficient function by listing its values $\tilde{Q}(\tilde{w}_j, \tilde{w}_{j-2})$ just for $\tilde{w}_j$ from the first quadrant $\tilde{\W}_0 \ni \tilde{w}_j$. All the rest can then be obtained by rotation:
\begin{equation*}
    \tilde{Q}((R_{\pi/2})^k \cdot (\tilde{w}_j, \tilde{w}_{j-2})) =  (R_{\pi/2})^k \cdot \tilde{Q}(\tilde{w}_j, \tilde{w}_{j-2}) \, .
\end{equation*}

For all $\left( \tilde{w}_j, \tilde{w}_{j-2} \right) \in \tilde{\W}^2$ with $\tilde{w}_j \in  \tilde{\W}_0$, the coefficients $\tilde{q}_j$ assigned by $\tilde{Q}$ are listed below:
\begin{itemize}
    \item $\tilde{Q} \left( (0, 0)^\top , \tilde{w}_{j-2} \right) := (0, 0)^\top$ for any $\tilde{w}_{j-2} \in \tilde{\W}$;
    \item $\tilde{Q} \left( (2, 0)^\top , \tilde{w}_{j-2} \right) := (0, 1)^\top$ for any $\tilde{w}_{j-2} \in \tilde{\W}$;
    \item $\tilde{Q} \left( (2, 2)^\top , \tilde{w}_{j-2} \right) := (-1, 1)^\top$ for any $\tilde{w}_{j-2} \in \tilde{\W}$;
    \item $\tilde{Q} \left( (1, 0)^\top , \tilde{w}_{j-2} \right) := (0, 0)^\top$ for $\tilde{w}_{j-2} = (b, c)^\top$ with $c \geq 0$,
        \item $\tilde{Q} \left( (1, 0)^\top , \tilde{w}_{j-2} \right) := (0, 1)^\top$ for $\tilde{w}_{j-2} = (b, c)^\top$ with $c < 0$;
    \item $\tilde{Q} \left( (1, 2)^\top , \tilde{w}_{j-2} \right) := (-1, 0)^\top$ for $\tilde{w}_{j-2} = (b, c)^\top$ with $c \geq 0$,
        \item $\tilde{Q} \left( (1, 2)^\top , \tilde{w}_{j-2} \right) := (-1, 1)^\top$ for $\tilde{w}_{j-2} = (b, c)^\top$ with $c < 0$;
    \item $\tilde{Q} \left( (2, 1)^\top , \tilde{w}_{j-2} \right) := (0, 1)^\top$ for $\tilde{w}_{j-2} = (b, c)^\top$ with $b \leq 0$,
        \item $\tilde{Q} \left( (2, 1)^\top , \tilde{w}_{j-2} \right) := (-1, 1)^\top$ for $\tilde{w}_{j-2} = (b, c)^\top$ with $b > 0$;
    \item $\tilde{Q} \left( (1, 1)^\top , \tilde{w}_{j-2} \right) := (0, 0)^\top$ for $\tilde{w}_{j-2} = (b, c)^\top$ with $b \leq 0$ and $c \geq 0$,
        \item $\tilde{Q} \left( (1, 1)^\top , \tilde{w}_{j-2} \right) := (-1, 1)^\top$ for $\tilde{w}_{j-2} = (b, c)^\top$ with $b \geq 0$ and $c \leq 0$, except for $b = c = 0$,
        \item $\tilde{Q} \left( (1, 1)^\top , \tilde{w}_{j-2} \right) := (-1, 0)^\top$ for $\tilde{w}_{j-2} = (b, c)^\top$ with $b > 0$ and $c > 0$,
        \item $\tilde{Q} \left( (1, 1)^\top , \tilde{w}_{j-2} \right) := (0, 1)^\top$ for $\tilde{w}_{j-2} = (b, c)^\top$ with $b < 0$ and $c < 0$.
\end{itemize}

By checking all the possible variants $\left( \tilde{w}_j, \tilde{w}_{j-2}, \tilde{w}_{j-4} \right) \in \tilde{\W}^3$, it can be verified that the final digit $\tilde{z}_j$ calculated by formulas \eqref{ParAd-M1_9-digits} and \eqref{ParAd-M2_9-digits} is always an element of the desired digit set $\tilde{\D}$.

Correct value of the final sum $\tilde{z} = \sum_{j \in \Z} M^j \tilde{z}_j$ is guaranteed by
\begin{eqnarray*}
    \tilde{z} & = & \sum_{j \in \Z} M^j \tilde{z}_j = \sum_{j \in \Z} M^j (\tilde{w}_j + \tilde{q}_{j-2} - M^2 \tilde{q}_j) = \sum_{j \in \Z} M^j \tilde{w}_j + \sum_{j \in \Z} M^j \tilde{q}_{j-2} - \sum_{j \in \Z} M^{j+2} \tilde{q}_j = \\
    & = & \tilde{w} + \sum_{j \in \Z} M^j \tilde{q}_{j-2} - \sum_{l \in \Z} M^l \tilde{q}_{l-2} = \tilde{w} + 0 = \tilde{w} \, .
\end{eqnarray*}
\end{example}

\begin{remark}\label{ingrid}
Consider $M \in \Z^{m\times m}$ with no eigenvalue on the unit circle. Then at least one eigenvalue~$\lambda$ of~$M$ satisfies $|\lambda| > 1$. The eigenvalue ~$\lambda$ is an algebraic integer, as it is a root of the characteristic polynomial $f \in \Z[X]$ of the matrix~$M$, and, obviously, $f$~is monic. If, moreover, $f$~is irreducible over~$\mathbb{Q}$, then there exists an isomorphism $\xi$ between $\Z^m$ and~$\Z[\lambda]$:
\begin{eqnarray}
    \label{Z-lambda}
    \Z[\lambda] & = & \{a_0 + a_1 \lambda + \cdots + a_{m-1} \lambda^{m-1} : a_k \in \Z \ \text{for all }k = 0, 1, \ldots, m-1\}; \\
    \label{isomorphism-Z-lambda}
    \xi: \Z[\lambda] \mapsto \Z^m & \quad & \xi(a_0 + a_1 \lambda + \cdots + a_{m-1}) := (a_0, a_1, \cdots, a_{m-1})^\top \, .
\end{eqnarray}

Consequently, any algorithm for parallel addition in the number system $(\lambda, \A)$  with $\A \subset \Z[\lambda]$ can be transformed by the isomorphism~$\xi$ to the matrix numeration system $(M, \xi(\A))$. In particular, if $\A \subset \Z$, then the digit set $\xi(\A) \subset \{c e_1: c \in \Z\}$, where $e_1 = (1, 0, \ldots, 0)^\top$. Any known result on the minimal size of the digit set allowing parallel addition in the number system with base~$\lambda$ can be applied to the matrix numeration system with base~$M$, thanks to irreducibility of~$f$ over~$\mathbb{Q}$. Some results on the minimal size of digit sets for parallel addition in number systems (with complex base $\beta \in \C$) can be found e.g. in~\cite{FrPeSv13}.
\end{remark}

Theorem \ref{veta} states that each non-singular matrix $M \in \Z^{m \times m}$ with no eigenvalue on the unit circle can be equipped with a suitable finite digit set $\D \subset \Z^m$ such that any integer vector $x \in \Z^m$ is representable in the system $(M, \D)$. As shown in~\cite{PelVa22}, the assumption $|\lambda| \neq 1$ for each eigenvalue~$\lambda$ of~$M$ is not necessary for representability of $\Z^m$ in $(M, \D)$. Nevertheless, this assumption is necessary for parallel addition in $(M, \D)$, as shown below.

\begin{proposition}\label{Ano}
If addition in the matrix numeration system $(M, \D)$ with base $M \in \Z^{m\times m}$ and (finite) digit set $\D \subset \Z^m$ is doable in parallel, then no eigenvalue of~$M$ lies on the unit circle.
\end{proposition}

\begin{proof}
By Lemma~\ref{extendDigits}, we consider, without loss of generality, that the digit set~$\D$ generates~$\R^m$ -- i.e., that $\R^m$~is the linear hull of~$\D$.

Assume, for contradiction, that $M$~has an eigenvalue $\lambda \in \C$ on the unit circle, $|\lambda| = 1$. Let $u \in \C^m$ be an eigenvector of the matrix~$M^\top$ to the eigenvalue~$\lambda$ -- i.e., $M^\top u = \lambda u$. As $\D$~generates~$\R^m$, the vector~$u$ cannot be orthogonal to all digits, hence there exists a digit $d \in \D$ such that  $\alpha u^\top d = 1$ for some $\alpha \in \C$. Let ~$v$ be the eigenvector $v = \alpha u$, so that $v^\top d = 1$ and $v^\top M = \lambda v^\top$.

Let parallel addition be performed by a~$p$-local function. Denote
\begin{equation*}
    S = \max\bigl\{\, \bigl| \sum_{j=0}^{p-1} \lambda^j \, v^\top d_j \bigr| \,  \, d_j \in \D \bigr\} \,.
\end{equation*}
As $|\lambda| = 1$, there exist infinitely many $j \in \N$ such that $\Re(\lambda^j) > \tfrac12$.  Hence one can find $N \in \N$ and coefficients $\varepsilon_0, \varepsilon_1, \ldots, \varepsilon_{N-1} \in \{0,1\}$ such that $\Re \bigl(\, \sum_{j=0}^{N-1} \varepsilon_j \lambda^j \bigr) > 2S$. Let $x = \sum_{j=0}^{N-1} M^j x_j$ with $x_j \in \D$ such that
\begin{equation*}
    |\Re(v^\top x)| = \max\Bigl\{ \bigl| \Re\Bigl(v^\top \bigl(\, \sum_{j=0}^{N-1} M^j d_j \bigr) \Bigr) \bigr| \, : \, d_j \in \D \Bigr\} = \max\Bigl\{ \bigl| \Re\Bigl( \sum_{j=0}^{N-1} \lambda^j \, v^\top d_j \Bigr) \bigr| \, : \, d_j \in \D \Bigr\} \, .
\end{equation*}
Since $1$ and $0$ belong to $\{v^\top d \, : \, d \in \D \}$, we have
\begin{equation}\label{ono}
    |\Re(v^\top x)| \geq \Re \bigl( \, \sum_{j=0}^{N-1} \varepsilon_j \lambda^j \bigr) > 2S .
\end{equation}
The $p$-local function used to add $x + x$ produces digits $z _j \in \D$ such that
\begin{equation*}
    x + x = \sum_{j=N}^{N+p-2} M^j z_j + \sum_{j=0}^{N-1} M^j z_j + \sum_{j=-p+1}^{-1} M^j z_j \, .
\end{equation*}
After multiplication of~$2x$ by the vector $v^\top$ from the left, we have
\begin{equation}\label{onco}
    2 \, v^\top x = \lambda^N \ \sum_{j=0}^{p-2} \lambda^j \ v^\top z_{N+j} + v^\top \bigl( \, \sum_{j=0}^{N-1} M^j \,  z_j \bigr) + \lambda^{-p} \sum_{j=1}^{p-1} \lambda^j \, v^\top z_{j-p} \, .
\end{equation}
The definitions of $S$ and~$x$ guarantee that
\begin{equation*}
    \bigl| \sum_{j=0}^{p-2} \lambda^j \ v^\top z_{N+j} \bigr| \leq S, \quad \bigl| \sum_{j=1}^{p-1} \lambda^j \, v^\top z_{j-p} \bigr| \leq S, \quad \text{and} \quad \bigl| \Re\Bigl(v^\top  \bigl( \sum_{j=0}^{N-1} M^j \, z_j \bigr) \Bigr) \bigr| \leq |\Re(v^\top x)| \, .
\end{equation*}
Using these inequalities and the triangle inequality, together with~\eqref{onco} and the fact that $|\Re (y)| \leq |y|$ for every $y \in \C$, and with $|\lambda| = 1$, we get
\begin{equation*}
    2 |\Re (v^\top x)| \leq |2 v^\top x| \leq 2S + |\Re (v^\top x)| \quad - \text{ which is a contradiction to \eqref{ono}}.
\end{equation*}

\end{proof}

\section{Eventually periodic representations with expansive matrix base}

T.~V\'avra in~\cite{Va21} shows, for any algebraic complex base~$\beta$ with $|\beta| > 1$, that there exists a suitable (finite) digit set $\A \subset \Z$ such that any $x \in \mathbb{Q}(\beta)$ has an eventually periodic expansion in this base, i.e., $x = \sum_{j=-\infty}^N a_j \beta^j$, where the sequence $(a_j)_{j \leq N}$ of digits from~$\A$ is eventually periodic. Looking for an analogy to this result in the matrix systems, we first have to give a meaning of the previous sum in the case when the number base~$\beta$ is replaced with a matrix base~$M$ and (integer) number digits~$a_j$ by (integer) vector digits~$d_j$.\\

If a matrix $M \in \Z^{m \times m}$ is expansive, then $M^{-1}$~is contractive and there exists a vector norm~$\norm{\cdot}_c$ in~$\R^m$ such that $\norm{M^{-1}} < 1$, where $\norm{\cdot}$~is the matrix norm induced by the vector norm~$\norm{\cdot}_c$, see~\cite{IsaacsonKeller}. Let us recall that for these two norms the following inequalities hold:
\begin{eqnarray}
    \label{norm1}
    \norm{Ax}_c & \leq & \norm{A} \, . \, \norm{x}_c \quad {\rm for \ every} \ x \in \R^m {\rm \ and \ } A \in \R^{m \times m}; \\
    \label{norm2}
    \norm{AB} & \leq & \norm{A} \, .\, \norm{B} \quad {\rm for \ every} \ A, B \in \R^{m \times m}.
\end{eqnarray}

If $(d_j)_{-\infty}^N$ is a sequence of (vector) digits from a finite set $\D \subset \R^m$, then the vector $\sum_{j=-\infty}^N M^j\, d_j$ is well defined, as the sequence $\Bigl( \, \sum\limits_{j=-n}^N M^j \, d_j \Bigr)_{n=0}^{+\infty}$ is a Cauchy sequence. Indeed, let us denote $r := \norm{M^{-1}} < 1$ and $D := \max\{ \norm{d}_c : d \in \D\}$. The triangle inequality implies for every $n, q \in \N$ that
\begin{equation*}
    \norm{\sum_{j=-n}^N M^j \, d_j - \sum_{j=-n - q}^N M^j \, d_j }_c = \norm{\sum_{j = -n-q}^{-n-1}  M^j \, d_j }_c = \norm{\sum_{j=n+1}^{n+q} \bigl( M^{-1}\bigr)^j \, d_{-j }}_c \leq  \sum_{j=n+1}^{n+q} r^j D \leq \frac{D r^{n+1}}{1-r} \, ,
\end{equation*}
so the value $\frac{D r^{n+1}}{1-r}$ can be made arbitrarily small for all $n > n_0$, with sufficiently large $n_0 \in \N$.\\

Consequently, if $M$~is expansive, then there exists $\lim\limits_{n\to +\infty} \sum_{j=-n}^N M^j d_j$, which may be denoted as $\sum_{j=-\infty}^N M^j d_j$. In  the remaining part of this chapter, we  focus on matrix numeration systems with $M$~being an expansive matrix.\\

\medskip

Matrix numeration systems with base~$M$ being an expansive matrix have been intensively studied since the work of Vince. The main focus of the research in this area is on the systems where each lattice point has a unique representation.
Recall that  a lattice in~$\mathbb{R}^m$ is the set of all integer combinations of $m$ linearly independent vectors. A~lattice numeration system can be formalised as follows (see \cite{GerLaKovAt}):

\begin{defi}
Let $\Lambda$ be a lattice,  $M : \Lambda \rightarrow \Lambda$ be a linear operator  (also called the base or radix) and let $\D$ be a finite subset of $\Lambda$ containing~$0$ (called the digit set).
The triplet $(\Lambda, M, \D)$ is called a generalised number system (GNS) if every element $x \in \Lambda$ has a unique finite representation of the form
$$x = \sum_{j = 0}^{N} M^j d_j \, ,$$
where $N \in \N$, $d_j \in \D$ for every $j = 0, 1, \ldots, N$ and $d_N$ is a non-zero digit (if $x \neq 0$).
\end{defi}

Characterisation of triplets $(\Lambda, M, \D)$ forming GNS seems to be a difficult  task. To quote a necessary condition, we recall that two elements (lattice points) $x, y \in \Lambda$ are congruent modulo~$M$ if they belong to the same residue class, i. e., if $(x - y) \in M \Lambda$. We denote this fact by $x \equiv_\Lambda y \pmod{M}$.

\begin{theorem}{}\cite{KovacsAttila03}
If $(\Lambda, M, \D)$ is a GNS, then the following holds:
\begin{enumerate}
    \item The operator $M$ is  expansive.
    \item The digit set $\mathcal{D}$ is a complete residue system modulo M.
    \item $\det(M-I) \not = \pm 1$.
\end{enumerate}
\end{theorem}

Note that  the number of residue classes modulo $M$ equals $|\det(M)|$, hence the GNS has exactly $\#\D = |\det(M)|$ digits.

A sufficient condition for GNS was provided by  L. Germ\'an and A. Kov\'acs in \cite{GerLaKovAt}.

\begin{theorem}{\cite{GerLaKovAt}} Let $M: \Lambda \mapsto \Lambda$ be a non-singular linear operator. If the spectral radius of the inverse operator $M^{-1}$ is less than $\frac12$, then there exists  a digit set $\D\subset \Lambda$ such that $(\Lambda, M, \mathcal{D})$ is a GNS.
\end{theorem}

Any lattice  $\Lambda \subset \R^m$ is just an image of $\Z^m$  under a non-singular linear map. Since a linear map does not change the GNS property, we can consider without loss of generality that $\Lambda = \Z^m$. A linear operator mapping  $\mathbb{Z}^m$ to $\mathbb{Z}^m$ corresponds to multiplication of integer vectors by a  matrix $M \in \Z^{m\times m}$. Further results are known regarding existence of GNS with special types of the digits sets:
\begin{itemize}
    \item If $\D = \{k(1, 0, \ldots, 0)^\top): k \in \Z, \ 0 \leq k < |\det(M)|\}$, we speak about {\it canonical systems}.
    \item If $\D=\{k(1, 0, \ldots, 0)^\top): k \in \Z, \ -\frac12 |\det(M)| < k \leq \frac12 |\det(M)|\}$, the digit set is called {\it symmetric}.
    \item If $\D$ contains the lattice point of the smallest norm from each residue class (modulo~$M$), then the digits set~$\D$ is said to be {\it dense}. The choice of the smallest lattice point is not necessarily unique.
    \item The {\it adjoint} digit set consists of those lattice points which belong to $|\det(M)| \, \cdot \, [-\frac12, \frac12)^m$.
\end{itemize}

Results on GNS with the special digits sets can be consulted in \cite{HuKo}.\\

\medskip

If we abandon the requirement of uniqueness for representation of vectors from~$\Z^m$, then the question on existence of a suitable digit set for a given expansive base~$M$ is much more simpler. In fact, using a sufficiently redundant digit set~$\D$ allows to represent all vectors in $\R^m$.

First, we focus on vectors from~$\R^m$ which have eventually periodic representation in the numeration system $(M, \D)$, i.e., we are interested in the set
\begin{equation}
    {\rm Per}_{\D}(M)= \Bigl\{ \sum_{j=-\infty}^N M^j d_j : N \in \N, \ d_j \in \D \text{ for each } j \leq N \text{ and } (d_j)_{-\infty}^N \text{ is eventually periodic} \Bigr\} \, .
\end{equation}

To work with eventually periodic representations, we exploit a well known fact about contractive matrices, namely that
\begin{equation*}
    (I-A)^{-1} = \sum_{j=0}^{+\infty} A^j \text{ for a contractive matrix } A \in \C^{m \times m} \, .
\end{equation*}

\begin{lem}\label{inclusion}
Let $M \in \Z^{m \times m}$ be an expansive matrix and $\D \subset \Z^m$. If $x \in {\rm Per}_{\D}(M)$, then $x \in \mathbb{Q}^m$.
\end{lem}

\begin{proof}
Firstly, we assume that an~$(M, \D)$-representation of~$x$ has the form $0 \bullet (d_{-1} d_{-2} \cdots d_{-p})^\omega$. It means that $x = M^{-1} d_{-1} + M^{-2} d_{-2} + \dots + M^{-p} d_{-p} + M^{-p} x$, which implies that
\begin{equation}\label{purely-periodic}
    x = (I_m - M^{-p})^{-1} (M^{-1} d_{-1} + M^{-2} d_{-2} + \dots + M^{-p} d_{-p}) \, .
\end{equation}
Since all elements of~$M$ are integer, it follows that all the matrix powers $M^{-j}$ with $j \in \N$ and $(I_m - M^{-p})^{-1}$ belong to~$\mathbb{Q}^{m \times m}$. Therefore, $x \in \mathbb{Q}^{m}$.

Now, let us assume that a representation of~$x$ is eventually periodic. Then there exists $k \in \Z$ such that $M^k x = z + y$, where $z$~can be written as $z = \sum_{j=0}^{N} M^j d_j$ and $y$ has the purely periodic form $y = 0 \bullet (d_{-1} d_{-2} \cdots d_{-p})^\omega$. Obviously, $z \in \Z^m$ and, by the previous argumentation, $y \in \mathbb{Q}^{m}$. Hence $x = M^{-k} (z+y)$ belongs to $\mathbb{Q}^{m}$ as well.

\end{proof}

To study the implication opposite to the previous Lemma \ref{inclusion}, we use the concept of parallel addition.

\begin{lem}\label{PerD-M_closed-under-addition}
If the digit set~$\D$ allows parallel addition on ${\rm Fin}_{\D}(M)$, then it allows parallel addition also on ${\rm Per}_{\D}(M)$, and the result of addition is again eventually periodic. It means that, for such a digit set~$\D$, the set ${\rm Per}_{\D}(M)$ is closed under addition.
\end{lem}

\begin{proof}
Let us explain this statement, assuming the parallel addition on ${\rm Fin}_{\D}(M)$ is a $p$-local function, with $p = r + 1 + t$, memory~$r$ and anticipation~$t$.

Let $\omega_x$, $\omega_y$ be the lengths of the periods of $(M, \D)$-representations of the summands $x$, $y$, and denote $\omega_{xy} := \omega_x \omega_y$. Clearly, both $x$ and $y$ have $(M, \D)$-representations with the same period $\omega_{xy}$, and the same holds for the $(M, \D + \D)$-representation of their (eventually periodic) sum $w = x + y$ calculated by pure summation of digits on each position separately.

Without loss of generality, we can assume that $\omega_{xy} \geq p$; otherwise just use a sufficiently big multiple of the period length $\omega_{xy}$. Then it is clear that, by applying the $p$-local function $\varphi : (\D + \D)^p \mapsto \D$ onto the $(M, \D + \D)$-representation of the interim sum~$w$ with period $(w_1, \ldots, w_{\omega_{xy}})^{\omega}$, the resulting $(M, \D)$-representation of the final sum~$z$ must be eventually periodic with period $\omega_{xy}$ as well:
\begin{equation*}
    z_j := \varphi (w_{j+t}, \ldots, w_j, \ldots, w_{j-r}) = \varphi (w_{j+t+\omega_{xy}}, \ldots, w_{j+\omega_{xy}}, \ldots, w_{j-r+\omega_{xy}}) =: z_{j+\omega_{xy}} \in \D \, ,
\end{equation*}
considering any $p$-tuple $(w_{j+t}, \ldots, w_j, \ldots, w_{j-r})$ from the eventually periodic part of~$w$.
\end{proof}

\begin{theorem}\label{periodic1}
Let $M \in \Z^{m \times m}$ be an expansive matrix. Then there exists a finite digit set $\D \subset \Z^m$ such that ${\rm Per}_{\D}(M)= \mathbb{Q}^m$.
\end{theorem}

\begin{proof}
By Theorem~\ref{veta} and Lemma~\ref{extendDigits}, there exists a finite digit set~$\D$ such that addition in $(M,\D)$ can be performed in parallel. Due to Lemma~\ref{inclusion}, it remains to prove only the inclusion $\mathbb{Q}^m \subset {\rm Per}_{\D}(M)$. We denote by~$e_s$ the vector from $\Z^m$ whose $s^{th}$~coordinate equals~$1$ and all other coordinates are zero. In the first step, we show that for each $q \in \N$ and each $s = 1, 2, \ldots, m$, the vector $\frac{1}{q} e_s$ has an eventually periodic $(M, \D)$-representation. For this purpose, we define a congruence relation on the matrices from $\Z^{m \times m}$.

We say that $A \in \Z^{m \times m}$ is congruent modulo~$q$ to $B \in \Z^{m \times m}$, if $(A-B) \in  q \, \Z^{m \times m}$. As the number of congruence classes is finite (for a fixed $q \in \N$), we find in the list $I, M, M^2, M^3, \ldots$ two matrices from the same congruence class. In other words, there exist $k, \ell \in \N, \ell > 0$ such that $M^{k+\ell} - M^k = q C$ for some $C \in \Z^{m \times m}$, or, equivalently
\begin{equation}\label{perioda}
    \tfrac{1}{q} \, I =  M^{-\ell - k} \, (I -M^{-\ell})^{-1} \, C = M^{-\ell - k} \Bigl( \, \sum_{j=0}^{+\infty} M^{-\ell j} \Bigr) \, C \, .
\end{equation}
Let $C_1$ denote the first column of the matrix~$C$. As $\Z^m \subset {\rm Fin}_{\D}(M)$, we can write $C_1 = \sum_{j=0}^n M^j f_j$. It is due to Lemma~\ref{dominant} and Remark~\ref{FractionalPart} that the bottom index of the sum equals zero, as $M$~is expansive, so all of its eigenvalues are $>1$~in modulus. The first column of the matrix equality~\eqref{perioda} then equals
\begin{equation*}
    \tfrac1{q} \, e_1 = M^{-\ell - k}\Bigl( \, \sum_{j=0}^{+\infty} M^{-\ell j} f_0 + \, \sum_{j=0}^{+\infty} M^{-\ell j+1} f_1 + \cdots + \, \sum_{j=0}^{+\infty} M^{-\ell j+n} f_n \Bigr) \, .
\end{equation*}

Thus, we have expressed $\tfrac1{q} \, e_1$  as a sum of $n+1$~vectors with eventually periodic $(M, \D)$-representation. Since ${\rm Per}_{\D}(M)$ is closed under addition, due to Remark~\ref{PerD-M_closed-under-addition}, the vector~$\tfrac1{q} \, e_1$ belongs to ${\rm Per}_{\D}(M)$ as well. Analogously, the same holds for~$\tfrac1{q} \, e_2, \ldots, \tfrac1{q}\, e_m$ and for $-\tfrac1{q} \, e_1, \ldots, -\tfrac1{q}\, e_m$, therefore they are also elements of~${\rm Per}_{\D}(M)$.

In the second step, we consider an arbitrary vector $x \in \mathbb{Q}^m$. We find $q \in \N$ and $y = (y_1, \ldots, y_m)^\top \in \Z^m$ such that $x = y_1 \bigl(\frac1{q} \, e_1\bigr) + \cdots + y_m \bigl(\frac1{q} \, e_m\bigr)$. It means that $x$~is a sum of $(|y_1 |+ \cdots + |y_m|)$ vectors from the set ${\rm Per}_{\D}(M)$, which is closed under addition. Hence $x \in {\rm Per}_{\D}(M)$.
\end{proof}

\begin{corollary}\label{total}
Let $M \in \Z^{m \times m}$ be an expansive matrix. Then there exists a finite digit set $\D \subset \Z$ such that every $x \in \R^m$ has an $(M,\D)$-representation.
\end{corollary}

\begin{proof}
Let $\norm{\cdot}_c$ be the vector norm of~$\R^m$ mentioned in \eqref{norm1} -- \eqref{norm2}, for which the induced matrix norm of the contractive matrix~$M^{-1}$ is $\norm{M^{-1}} = r <1$.

Since $M$ is expansive, any vector $y \in \Z^m$ has an $(M,\D)$-representation in the form $y = \sum_{j=0}^N M^j d_j$ for some $N \in \N$ and $d_j \in \D$. In general, the representation is not unique. We denote by ${\rm height}(y)$ the minimal $N \in \N$ among  all the $(M, \D)$-representations of~$y$.

Let $x \in \R^m $ and $\bigl(x^{(n)}\bigr)_{n \in \N}$ be a sequence of vectors from~$\mathbb{Q}^m$ such that $\lim\limits_{n\to \infty} x_n = x$. By the previous theorem, we have $x^{(n)} = \sum_{j=-\infty}^{N_n} M^j d^{(n)}_j$. Denote the integer and fractional parts of~$x^{(n)}$ by $y^{(n)} := \sum_{j=0}^{N_n} M^j d^{(n)}_j$ and $z^{(n)} := \sum_{j=-\infty}^{-1} M^j d^{(n)}_j$, respectively. Obviously, $y^{(n)} \in \Z^m$. Assume that $(M, \D)$-representations of the integer parts~$y^{(n)}$ satisfy $N_{(n)} = {\rm height}(y^{(n)})$ for every $n \in \N$. The size of the fractional parts~$z^{(n)}$ is bounded. Indeed, $\norm{z^{(n)}}_c = \norm{\sum_{j=-\infty}^{-1} M^j d^{(n)}_j}_c \leq \sum_{j=-\infty}^{-1} r^{-j} D = \frac{rD}{1-r}$. Since any convergent sequence is bounded, the integer parts $y^{(n)} \in \Z^m$ are bounded as well, as $\norm{y^{(n)}}_c \leq \norm{x^{(n)} - z^{(n)}}_c \leq \norm{z^{(n)}}_c + \norm{x^{(n)}}_c$. It means that $y^{(n)}$~can take only finitely many values in~$\Z^m$, and thus their heights~$N_{(n)}$ are bounded as well, say by~$N$.

Hence we can rewrite the representation of~$x^{(n)}$ for each $n \in \N$ into the form $x^{(n)} = \sum_{j=-\infty}^{N} M^j d^{(n)}_j$ with the uniform upper index~$N$ of the sums (if necessary, we add leading zero coefficients to sums). Now we are ready to find an $(M, \D)$-representation of $x = \lim\limits_{n\to \infty}x^{(n)}$.

We construct the sequence $d_N d_{N-1} \cdots d_0 \bullet d_{-1} d_{-2}\cdots$ of digits from~$\D$. A~digit which appears infinitely times among $x_N^{(n)}$ will be chosen as~$d_N$. Then, we chose $d_{N-1}$ as such a digit that the pair of digits $d_{N} d_{N-1}$ appears infinitely many times among $x_N^{(n)} x^{(n)}_{N-1}$. Similarly, $d_{N-2}$ is chosen as such a digit that the triplet  $d_{N} d_{N-1} d_{N-2}$ appears infinitely many times among $x_N^{(n)} x^{(n)}_{N-1} x^{(n)}_{N-2}$, and so on. By this construction, for every $l \in \N$ there exists $k_l \in \Z$ such that the strings $d_N d_{N-1} \cdots d_0 \bullet d_{-1} d_{-2} \cdots$ and $x^{(k_l)}_{N} x^{(k_l)}_{N-1} \cdots x^{(k_l)}_{0} \bullet x^{(k_l)}_{-1} x^{(k_l)}_{-2}\cdots$ have a common prefix of length at least~$l$. Therefore,
\begin{equation*}
    \norm{\sum_{j=-\infty}^N M^j d_j - \sum_{j=-\infty}^N M^j x^{(k_l)}_j} = \norm{\sum_{j=-\infty}^{N-l} M^j \bigl(d_j - x^{(k_l)}_j \bigr)} \leq 2D \sum_{j={l-N}}^{+\infty} \, r^{j} = \frac{2D}{1-r} \, r^{l-N} \, .
\end{equation*}
Since the right side of this inequality tends to~$0$ with $l \rightarrow +\infty$, we obtain
\begin{equation*}
    \sum_{j=-\infty}^N M^j d_j = \lim_{l \to +\infty} x^{(k_l)} = \lim_{n \to +\infty} x^{(n)} = x \, ,
 \end{equation*}
and thus $\sum_{j=-\infty}^N M^j d_j$ is an $(M,\D)$-representation of~$x$.

\end{proof}

\section{Open questions}
We have focused only on two  questions connected with $(M,\D)$-representation of vectors: computability of addition in parallel and eventually periodic representations. Many open problems still remain unresolved. Let us list some of them:
\begin{enumerate}
    \item What is the minimal size of a digit set $\D \subset \Z^m$ with the property $\Z^m \subset {\rm Fin}_{\D}(M)$? This question was already tackled in~\cite{CaHaVa22} for very special matrices, namely for Jordan blocks~$J_m(1)$ corresponding to the eigenvalue~$1$.
    \item Is it possible for a given non-singular  matrix $M \in \Z^{m\times m}$ to find a (finite) digit set $\D \subset \Z^m$ such that every element in ${\rm Fin}_{\D}(M)$ has a unique representation in this numeration system? If $M$ is expansive, then size of the suitable digit set (if it exists) is $|\det(M)|$, see \cite{KovacsAttila03}. What is an analogy of this result for a  non-expansive matrix?
    \item Does there exist any (finite) digit set such that multiplication of vectors from~$\R^m$ by a~scalar $x \in \R$ can be performed by an on-line algorithm? Note that K.~Trivedi and M.~Ercegovac in~\cite{TrEr} designed on-line algorithms for multiplication and division of two numbers represented in a numeration system $(\beta, \A)$ with $\beta \in \N$. Their algorithms were later generalised to numeration systems $(\beta, \A)$ with base~$\beta$ being a (real or complex) Pisot number~\cite{FRPaPeSv}.
\end{enumerate}

\section*{Acknowledgements}

Edita Pelantov\'{a} acknowledges financial support by The Ministry of Education, Youth and Sports of the Czech Republic, project no. {CZ.02.1.01/0.0/0.0/16\_019/0000778}.





\begin{thebibliography}{11}

\bibitem{AkDrJa12} S. Akiyama, P. Drungilas, J. Jankauskas, Height reducing problem on algebraic integers, Funct. Approx. Comment. Math. {\bf 47} (2012) 105--119.

\bibitem{Avizienis} A. Avizienis, Signed-digit number representations for fast parallel arithmetic, IRE Trans. Electron. Comput. {\bf 10} (1961) 389--400.

\bibitem{CaHaVa22} J. W. Caldwell, K. G. Hare, T. Vávra, Non-expansive matrix number systems with bases similar to $J_n(1)$, arXiv: 2110.11937.

\bibitem{Ca1840} A. Cauchy, Sur les moyens d’\'eviter les erreurs dans les calculs num\'eriques, C.R. Acad. Sc. Paris s\'erie I {\bf 11} (1840) 789-–798.

\bibitem{FRPaPeSv} Ch. Frougny, M. Pavelka, E. Pelantov\'a, M. Svobodov\'a, On-line algorithms for multiplication and division in real and complex numeration systems,  Discr. Math. Theor. Comput. Sci. {\bf 21}(3) (2019).

\bibitem{FrPeSv13} Ch. Frougny, E. Pelantov\'a, M. Svobodov\'a, Minimal digit sets for parallel addition in non-standard numeration systems, Jour. Integ. Seq. {\bf 16}(2) (2013).

\bibitem{FrPeSv} Ch. Frougny, E. Pelantová, M. Svobodová, Parallel addition in non-standard numeration systems, Theor. Comp. Sci. {\bf 412} (2011) 5714--5727.

\bibitem{GerLaKovAt} L. Germ\'an, A. Kov\'acs, On number system constructions, Acta Math. Hungar. {\bf 115}(1-2) (2007) 155--167.

\bibitem{HuKo} P. Hudoba, A. Kovacs, Toolset for Supporting the Research of
Lattice Based Number Expansions, Acta Cybernetica 25 (2021) 271--284.

\bibitem{Gr1885} V. Gr\"unwald, Intorno all’aritmetica dei sistemi numerici a base negativa con particolare riguardo al sistema numerico a base negativo-decimale per lo studio delle sue analogie coll’aritmetica ordinaria (decimale), Giornale di matematiche di Battaglini {\bf 23} (1885)  203--221.

\bibitem{JaTh18} J. Jankauskas, J. M. Thuswaldner, Characterization of rational matrices that admit finite digit representations, Lin. Alg. \& App. {\bf 557} (2018) 350-358.

\bibitem{IsaacsonKeller} E. Isaacson, H. B. Keller, Analysis of numerical methods, John Wiley \& Sons (1966).

\bibitem{KiTh14} P. Kirschenhofer, J. M. Thuswaldner, Shift radix systems - a survey, Proceedings of Numeration and substitution 2012, Res. Inst. Math. Sci. (RIMS), Kyoto, (2014) 1--59.

\bibitem{Kn60} D. Knuth, An imaginary number system, Communic. ACM {\bf 3} (1960) 245--247.

\bibitem{KovacsAttila03} A. Kov\'acs, Number expansions in lattices, Math. Comput. Modelling {\bf 38} (7-9) (2003) 909--915.

\bibitem{KoPe91} B. Kov\'acs, A. Peth\"o, Number systems in integral domains, especially in orders
of algebraic number fields, Acta Sci. Math Szeged {\bf 55} (1991) 287--299.

\bibitem{Leg} J. Legerský, Minimal non-integer alphabets allowing parallel addition, Acta Polytech., {\bf 58}(5) (2017).

\bibitem{LegSvo} J. Legerský, M. Svobodov\'a, Construction of algorithms for parallel addition in expanding bases via extending window method, Theor. Comp. Sci. {\bf 795}(5) (2019).

\bibitem{LM} D. Lind, B. Marcus, An introduction to symbolic dynamics and coding, Cambridge University Press (1995).

\bibitem{PelVa22} E. Pelantov\'a, T. V\'avra,  On positional representation of integer vectors, Lin. Alg. \& App. {\bf 633} (2022) 316--331.

\bibitem{Pen65} W. Penney, A 'binary' system for complex numbers, Jour. ACM {\bf 12} (1965) 247-–248.

\bibitem{Re57} A. R\'enyi, Representations for real numbers and their ergodic properties, Acta Math. Acad. Sci. Hung. {\bf 8} (1957) 477--493.

\bibitem{Sch80} K. Schmidt, On periodic expansions of Pisot numbers and Salem numbers, Bull. London Math. Soc. {\bf 12} (1980) 269--278.

\bibitem{TrEr} K. S. Trivedi, M. D. Ercegovac, On-line algorithms for division and multiplication,
IEEE Transactions on Computers {\bf C-26} (1977) 681-–687.

\bibitem{VaVe18} T. V\'avra, F. Veneziano, Pisot unit generators in number fields, Jour. Symb. Comput. {\bf 89} (2018) 94--108.

\bibitem{Vince93a} A. Vince, Radix representation and rep-tiling, Proceedings of the 24-th Southeastern International Conference on Combinatorics, Graph Theory, and Computing (Boca Raton, FL, 1993), Congr. Numer. {\bf 98} (1993) 199-–212.

\bibitem{Vince93b} A. Vince, Replicating tessellations, SIAM Jour. Discr. Math. {\bf 6}(3) (1993) 501--521.

\bibitem{Va21} T. V\'avra, Periodic representations in Salem bases, Israel Jour. Math. {\bf 242} (2021) 83-–95.

\end{thebibliography}
\end{document}